\documentclass[11pt]{article}

\usepackage{amsmath,amsthm,amsfonts,amscd,latexsym,amssymb,amsbsy,dsfont,mathrsfs,color}
\usepackage[small, centerlast, it]{caption}
\usepackage{epsfig}
\usepackage[titles]{tocloft}
\usepackage[latin1]{inputenc}
\usepackage{verbatim} 
\usepackage[active]{srcltx} 
\usepackage{fancyhdr}
\usepackage{caption}
\usepackage{enumerate}
\usepackage{color}
\usepackage{hyperref}

\definecolor{NiColor}{RGB}{77,77,255}
\definecolor{NiColoRed}{RGB}{255,77,77}
\definecolor{NiCitation}{RGB}{77,255,77}

\newsymbol\rest 1316    
\def\emptyset{\varnothing} 
\def\b1{{1\!\!1}}

\def\bC{{\mathbb C}}           

\def\bR{{\mathbb R}}

\def\bZ{{\mathbb Z}}

\def\beq{\begin{eqnarray}}
\def\eeq{\end{eqnarray}}

\usepackage{sectsty}
\sectionfont{\normalsize}
\subsectionfont{\normalsize\normalfont\itshape}
\newtheoremstyle{TheoremStyle}
{3pt}
{3pt}
{\slshape}
{}
{\bf}
{:}
{.5em}
{}

\theoremstyle{TheoremStyle}
\newtheorem{theorem}{Theorem}
\newtheorem{corollary}[theorem]{Corollary}
\newtheorem{proposition}[theorem]{Proposition}
\newtheorem{lemma}[theorem]{Lemma}
\newtheorem{definition}[theorem]{Definition}
\newtheorem{remark}[theorem]{Remark}
\newtheorem{example}[theorem]{Example} 

\begin{document}


\par
\bigskip
\large
\noindent
{\bf On  the global Hadamard parametrix in QFT   and  the signed squared geodesic distance  defined   in domains  larger than convex normal neighbourhoods}
\bigskip
\par
\rm
\normalsize


\noindent {\bf Valter Moretti}\\
\par

\noindent 
  Department of  Mathematics, University of Trento, and INFN-TIFPA \\
 via Sommarive 14, I-38123,  Povo (Trento), Italy.\\
 valter.moretti@unitn.it\\

 \normalsize

\par

\rm\normalsize

\noindent {\small August,  2021}

\rm\normalsize


\par

\begin{abstract} \noindent We consider the  global Hadamard condition 
and the notion of Hadamard parametrix  whose use  is pervasive in algebraic QFT in curved spacetime (see refences in the main text).
We point out the existence of a technical  problem in the literature
concerning  well-definedness  of  the  global Hadamard parametrix in normal neighbourhoods of Cauchy surfaces.
We discuss in particular  the definition of the (signed) geodesic distance $\sigma$  
and related structures  in an open  neighbourhood of the diagonal of $M\times M$ larger than  $U\times U$, for a normal convex neighborhood $U$, where $(M,g)$ is a Riemannian or Lorentzian (smooth Hausdorff paracompact) manifold.
We eventually  propose a quite natural solution  which slightly changes the original definition by B.S.  Kay and R.M. Wald and relies  upon  some non-trivial consequences of the paracompactness property. The proposed  re-formulation is in agreement with M.J. Radzikowski's  microlocal  version of the Hadamard condition.
\end{abstract}

\section{Introduction}\label{Sec: Introduction}
The use of {\em Hadamard states} is nowadays pervasive in {\em algebraic QFT}  (aQFT)  in curved spacetime (see,  e.g.,  \cite{BF96,BF00,HW,HW02,Stress,Mo08,DMP2,DMP,DFP08,Sanders,FS08,FV03,FV13},  and   \cite{book} for a recent survey on aQFT).
The rigorous definition of Hadamard state in terms of {\em short-distance behaviour} of the two-point function  was stated  in  the celebrated paper \cite{KW} by B.S. Kay and R.M. Wald for the first time. Some years later, that technically complex definition was translated into the language of microlocal analysis within  a pair of nice papers by M.J. Radzikowski \cite{Rad,Rad2}.

 The original geometric definition of \cite{KW} of a global {\em Hadamard parametrix} has been  exploited for instance  to deal with rigorous intepretations of the Hawking radiation,
see    \cite{MP} and the  recent  interesting paper \cite{KNV}. Just to mention some other  applications of the Hadamard parametrix in aQFT (the following  list of examples is by no means exhaustive) we can say that 
it plays a crucial role  in the definition of {\em locally-covariant Wick powers} \cite{HW,KM}, including  the definition of  the {\em stress-energy tensor operator}  \cite{Stress}. 
The Hadamard  parametrix has been also employed in the study of  {\em quantum energy inequalities} \cite{FS08}.  
It has been also used in {\em semiclassical approaches} to the quantum gravity and   {\em cosmological applications} \cite{DFP08,MPS}. 

Locally (and a bit roughly) speaking, in a globally-hyperbolic four-dimensional spacetime $(M,g)$, an algebraic state $\omega$ of a real scalar  Klein-Gordon quantum field 
is of {\em Hadamard type} if its {\em two-point function} $\Lambda_\omega$ has the {\em Hadamard short-distance singularity},
\beq \Lambda_\omega(x,y) = \frac{1}{(2\pi)^2}\left(\frac{\:\:\:\Delta(x,y)^{1/2}}{\sigma(x,y)} + v(x,y) \ln \sigma(x,y) \right)+ H_\omega(x,y)\label{fH}\eeq
when viewed as an integral kernel
(see Section \ref{secHcond} for some technical details here disregarded).  $H_\omega$ is a smooth function depending on  the state $\omega$, whereas $\Delta, v, \sigma$ are universal geometric objects constructed out of the local geometry only. In particular, $\sigma(x,y)$ is the so-called {\em signed squared geodesic distance} of $x,y\in M$. It is defined as the squared length  -- with the appropriate sign -- of {\em the} geodesic segment joining $x$ and $y$.  The so-called {\em Hadamard parametrix} is the singular universal part $\Lambda_\omega(x,y)- H(x,y)$.

Since there are many geodesics, in principle, joining $x$ and $y$, a standard possibility is to assume  that the identity above is true in a {\em normal convex neighbourhood}\footnote{A weaker requirement is that the identity is valid in a normal neighborhood of one of the points.} (see Section \ref{SecN}). This  is  an open set $U$ such that every pair of points $x,y\in U$ can be joined by a {\em unique} geodesic segment
$\gamma: [0,1] \to M$  {\em  that belongs to  the set}: $\gamma([0,1]) \subset U$. This elementary precaution is not enough however in the global definition discussed in \cite{KW}. It is because  (\ref{fH}) is  assumed  to be valid for  pairs $(x,y)$ contained in   {\em many} normal convex neighbourhoods. In principle this  gives rise to a cumbersome many-valued function  $\sigma$. This is one of the most difficult  technical issues tackled   in \cite{KW}. 

\begin{remark}\label{REM0} {\em If the spacetime points   $x$ and $y$ belong to  many  convex normal neighbourhoods,
different choices of such a neighbourhood may not only lead to different values for $\sigma(x,y)$ around $(x,y)$, but also to different singularity structures as determined by the Hadamard parametrix
(\ref{fH}). Therefore, in principle,  the introduced   issue   not only affects the construction of the Hadamard
parametrix, but also the singularity structure that the Hadamard condition is assumed  to describe. In particular, Hadamard states may not satisfy (\ref{fH}) in every convex normal 
neighbourhood. (See also the end of Example \ref{EX}.)}
\end{remark}

If $x$ and $y$ are {\em causally related},  a natural choice of $U$ for the given $x,y$  exists which solves the problem of the definition of $\sigma(x,y)$ and it was adopted in \cite{KW} (see Remark \ref{remJ} below for more details). $U$, if it exists,  is a normal convex neighbourhood  that {\em simultaneously includes $x,y$ and their causal  double cone $J(x,y)$} (see (\ref{JJ23})). Obviously, every other normal convex neighbourhood $U'$ which  both includes $x,y$ and $J(x,y)$ must also contain  the  causal geodesic $\gamma$ joining $x $ and $y$ in $U$. As $U'$ is convex, $\gamma$ is also {\em its own} geodesic $\gamma'$ joining $x$ and $y$ in $U'$. There is, in fact,  {\em only one} geodesic segment (parametrized in $[0,1]$) joining a pair of causally related points $x$ and $y$ in common for the subfamily of the said normal convex sets. For these pairs $(x,y)$, $\sigma(x,y)$ can be therefore unambiguously defined. 

Physics is properly reflected by the family of  causal geodesics, but a mathematically coherent definition of the Hadamard parametrix needs to consider also non-causal geodesics: the non-causal ones  ``arbitrarily close'' to  the  causal ones. For technical reasons,  in \cite{KW}  $\sigma$ was therefore  also required  to be smooth and well-defined  {\em in a neigbourhood} $\mathscr{O}$ of that special family of causally related pairs $(x,y)$. 
 We stress that the neighbourhood  $\mathscr{O}$ must also contain non-causally connected pairs.  The argument used to give a non-ambiguous definition of $\sigma$  {\em cannot} be used for those pairs. The existence of $\mathscr{O}$ with a non-ambiguous extension of the definition of $\sigma$ was   assumed  in \cite{KW} and also in \cite{Rad,Rad2} without a proof. In this author's view,  it remains a gap in the whole construction.
This work is devoted to that gap. 

  Referring to Remark \ref{REM0},  one eventually sees that the problem of the definition of $\sigma$
for non-causally connected pairs actually affects the definition of the
Hadamard parametrix in (\ref{fH}), but not its singularity structure (see Remark  \ref{remFIN}).

 We  shall not try to directly prove  the existence of that $\mathscr{O}$. Our solution  relies on a thin refinement of the definition of  Hadamard parametrix which is possible  thanks to a consequence of the paracompactness property. The final new  definition of Hadamard state, which is a quite slight modification of the original definition in \cite{KW}, though it is based on a non-trvial topological result, turns out to be in agreement with the microlocal version of the Hadamard condition.  

To achieve our  final goal,  in the first part of the paper, we shall focus  on the more abstract and mathematically-minded problem of a  well-posed  definition of $\sigma$ (and related geometric objects)  in a neighbourhood of the diagonal of $M\times M$. This issue is the core of the problem with the Hadamard parametrix, but it may have other applications in mathematical physics, so that it deserves a separate study. 

\begin{example}\label{EX}
{\em A concrete  elementary illustration of the problems one faces  when trying to define $\sigma(x,x')$ in a non-trivial spacetime  is the following one\footnote{This illustration is a straightforward  re-elaboration of an example proposed by  C.J. Fewster to the author  in  a private communication.}.
Consider the spacetime $(M,g)$ constructed out of the  $1+1$ Minkowski spacetime periodically identified under  $(t,x) \mapsto (t,x+2L)$ ($c=1$).  Differently from the Minkowski space, $M$ is not normal convex in its own right. {\em To define $\sigma(x,x')$, one is therefore forced to  make a choice of a normal convex open set containing $x$ and $x'$.}
 
Let $x =(0,L/2)$ and $x'=(L/2,L) \equiv (L/2, -L)$. These points are causally related and $J(x,x')$ (a null line segment) can be thickened up to become a normal convex neighbourhood $U$.
We wish to define the function $\sigma$ near $(x,x')$.
Let us first consider nearby points that are still causally connected.
In particular, let $y=x$ and $y' = (L/2 , L-\epsilon) = (L/2, -L-\epsilon)$  with $0<\epsilon \ll L$. These are causally related and near to $(x,x')$. In the considered case, we can  also assume, enlarging $U$ if necessary, that $U \supset J(y,y')$.
 We can create  further normal convex neighbourhoods which include $(y,y')$ by
\begin{itemize}

\item[(i)] thickening the line segment (in $\bR^2$) between $(0,L/2)$ and $(L/2, L-\epsilon)$  which is a timelike geodesic between $y$ and $y'$.
It produces a convex neighbourhood $U'$ which we can assume to satisfy again $U' \supset J(y,y')$. 
The line segment between $(0,L/2)$ and $(L/2, L-\epsilon)$  belongs to both $U$ and $U'$.  Consequently,  $\sigma(y,y')$ referred to $U$ must coincide with $\sigma(y,y')$ referred to $U'$.

\item[(ii)] thickening the line segment between $(0,L/2)$ and $(L/2, -L-\epsilon)$ which is a spacelike geodesic between $y$ and $y'$.  In this case we obtain a value for $\sigma(y,y')$ different from the one computed in $U$.
\end{itemize}
This illustrates why, when defining $\sigma(y,y')$ in the causally connected case, the condition on the geodesically convex neighbourhood that it contains $J(y,y')$ permits to select a common notion of  distance.

Next, consider the causally disconnected case.   Let $y=x=(0,L/2)$ and $y' = (L/2 , L+\epsilon) = (L/2, -L+\epsilon)$ with $0<\epsilon \ll L$. These are not causally connected.
We can still create normal convex neighbourhoods containing $y$ and $y'$ by thickening the line segment between  $(0,L/2)$ and $(L/2 , L+\epsilon)$ or the one 
between $ (L/2, -L+\epsilon)$ and  $(0,L/2)$ giving spacelike geodesics of differing lengths and different values of $\sigma(y,y')$.
There are actually infinitely many other possibilities that can be obtained using other image points $(L/2,(2n+1)L+\epsilon)$, $n \in \bZ$,  throughout a very slim thickening of these segments, wrapping on the cylinder without any self-intersection.

An explicit  illustration of the content of Remark \ref{REM0}  arises by referring to the example  above
and considering the timelike related points $p = (0, 0)$ and $p_0 = (2L, 0)$, which can also be
connected by null geodesics. Thickening a timelike or a null geodesic from $p$ to $p_0$ to a convex
normal neighbourhood leads to distinct singularity structures.}
\end{example}

\section{Extension of the signed squared geodesic distance and related structures}\label{SecN}
Smooth manifolds are hereafter assumed to be Hausdorff and paracompact\footnote{For a topological space locally homeomorphic to $\bR^n$ and equipped with a smooth atlas,
(a) the existence of a smooth Lorentz metric implies paracompactness (this is standard in the Riemannian case and it remains true with a more elaborated proof in the indefinite case), but not Hausdorffness;
(b) when the space is connected, 2nd countability and paracompactness are equivalent.
Thus, the usual requirements for spacetimes (existence of a Lorentz metric
and connectedness) imply that all the subtle required topological conditions
hold (paracompactness, 2nd countability, Lindel\"off) but Hausdorfness, which
must be imposed explicitly.}.  We adopt   the Lorentzian  signature $(-,+, \cdots, +$) and we follow \cite{ONeill} concerning basic defintions,  notation and results  in  the theory of Lorentzian manifolds  (see \cite{Min} for  an up-to-date general review).
\subsection{Normal convex sets and the local definition of the (signed) squared geodesic distance $\sigma$}

Following Chapter 5 of  \cite{ONeill},  $\exp :  \mathcal{D}\subset  TM \to M$ will denote the standard {\em exponential map} associated to the geodesic flow of $g$ of a smooth Riemannian or Lorentzian manifold $(M,g)$. Its  maximal domain   $\mathcal{D}$ is a fiberwise starshaped  open neighborhood of the zero section of $TM$,
and $\exp_p(v) := \exp(p,v)$ if $(p,v) \in  \mathcal{D}$. \\
\begin{definition}\label{defconvex}
{\em If $(M,g)$ is a Riemannian or Lorentzian smooth manifold, a 
 {\bf normal convex neighbourhood}  $U$ --  also known as 
{\bf normal convex open set} --  is an open  set $U\subset M$ such that, for every $q\in U$, 
there is a starshaped open neighbourhood $V^{(U)}_q$ of the origin of $T_qM$  such  that $\exp_q : V^{(U)}_q \to U$  is a diffeomorphism.}\\
\end{definition}

\begin{remark}\label{rem1}
{\em \begin{itemize}
\item[(1)] $U$ as above is {\bf geodesically starshaped}  with respect to every $p\in U$: for every other $q\in U$ there is only one geodesic segment $\gamma^{(U)}_{pq} : [0,1] \to M$ such that both $\gamma^{(U)}_{pq}(0)=p$, $\gamma^{(U)}_{pq}(1)=q$ and $\gamma^{(U)}_{pq}([0,1]) \subset U$ are valid (Proposition 31 in chapter 3 of \cite{ONeill}). By definition of $\exp_p$, it holds $\gamma^{(U)}_{pq}(t) = \exp_p(tv_q)$ where 
$v_q := \exp_p^{-1}(q)$ and $t\in [0,1]$.
\item[(2)] As  $U$ is geodesically starshaped with respect to every point $p\in U$,  it is not difficult to prove that  $V^{(U)}_p$ in Def. \ref{defconvex} is completely determined by $U$ and $p\in U$. 

\item[(3)] The set  $\cup_{p\in U} \{p\}\times V^{(U)}_p \subset \mathcal{D}$ is open in $TM$. This is   because the differential of the bijective map $U\times U \ni (p,q) \mapsto (p, (\exp_p|_{V_p^{(U)}})^{-1}(q)) \in \cup_{p\in U} \{p\}\times V^{(U)}_p$ is everywhere non-singular so that the map is open.  $\exp : \cup_{p\in U} \{p\}\times V^{(U)}_p \to U\times U$ is the inverse diffeomorphism and thus it is smooth (see in particular Lemma 9, Chap 5 \cite{ONeill} and comments around it).\\
\end{itemize}}
\end{remark}

A crucial result by Whitehead\footnote{The definition of normal convex neighbourhoods and Whitehead's result are more  generally true for smooth manifolds equipped with smooth affine connections \cite{KN}, however in this paper we stick to the smooth Levi-Civita connection generated by $g$.}
 proves that  (Propositition 7 Chapter 5 of \cite{ONeill} and its  proof.)\\

\begin{theorem}\label{teo01} For a Riemannian or Lorentzian smooth manifold $(M,g)$,  the family of  normal  convex open sets is not empty and forms a topological basis of $M$.\\
\end{theorem}

Among other important constructions, the Whitehead theorem and the properties of $\exp$ allow one to define the so-called 
{\bf (signed) squared geodesic distance} also known as
{\em  Synge's world function}. If $U$ is an open  normal convex set in $(M,g)$, 
\beq\label{defsigma}
\sigma_U(p,q) :=  g_p\left( \dot{\gamma}^{(U)}_{pq}(0),  \dot{\gamma}^{(U)}_{pq}(0)\right)
= \pm\left( \int_0^1 \sqrt{\left|g\left( \dot{\gamma}^{(U)}_{pq}(t),  \dot{\gamma}^{(U)}_{pq}(t)\right)\right| }dt\right)^2 \quad \mbox{for}\quad  p,q \in U
\eeq 
where the sign $-$ appears only if $g$ is Lorentzian and $\gamma^{(U)}_{pq}$ is timelike.
This function is  smooth in $U\times U$ 
because $\gamma^{(U)}_{pq}(t) = \exp_p\left(t (\exp_p|_{V_p^{(U)}})^{-1}(q)\right)$ is smooth in $[0,1]\times U\times U$, it being the restriction to $[0,1]\times U\times U$ of the  composition of three smooth functions (defined on open sets):  $\exp : \mathcal{D} \to M$, the multiplication with $t$, and a component of  $U\times U \ni (p,q) \mapsto (p, (\exp_p|_{V_p^{(U)}})^{-1}(q)) \in \cup_{p\in U} \{p\}\times V_p$ according to  (3) in Remark \ref{rem1}.

 It is evident that $\sigma_U(p,q)$ strictly depends on the choice of the normal convex neighborhood containing the points $p,q$. If there were another normal convex neighborhood $U'\ni p,q$, in general $\sigma_U(p,q) \neq \sigma_{U'}(p,q)$ because the two sides refer to generally different geodesic segments: one stays in $U$ and the other stays in $U'$, though both geodesics join $p$ and $q$.
 This fact prevents one  from defining $\sigma$ as a global smooth function  over $M\times M$. 

\subsection{Assignment of geodesics around the diagonal of $M\times M$ and extension of $\sigma$ thereon}\label{SectS}
A natural   issue which pops out at this juncture  is
{\em whether or not $\sigma$ can be more globally defined,  at least in an open neighbourhood $\mathscr{A}$ of the {\bf diagonal}
$\Delta_M := 
\{(p,p) \:|\: p \in M\}$ of 
  $M\times M$.} 

The root of the problem is that, generally speaking,  there are many geodesics connecting a pair of points $p,q$ and $\sigma(p,q)$ depends on the choice of one of those curves. One  restricts  to work in a ``small'' neighbourhood $\mathscr{A}$ of the diagonal of $M\times M$ because it seems that the choice should be easier if $p$ and $q$ are close to each other. (There are however results concerning really global definitions of $\sigma$, on the whole $M\times M$, when assuming suitable hypotheses on the topology of $M$ \cite{MC}.)
To address the issue above,  one may therefore wonder {\em if it is possible to define a jointly  smooth assignment of geodesic segments 
$\gamma_{pq}(t) = \Gamma(t,p,q)$ where  $t\in [0,1]$, $(p,q)$ varies in a neighborhood $\mathscr{A}$ of $\Delta_M$ and $\gamma_{pq}(0)=p$, $\gamma_{pq}(t)=q$.} Indeed, equipped with such an assignment, $\sigma$ can be defined on $\mathscr{A}$ by direct use of (\ref{defsigma}).\\
%

\begin{remark}  {\em If $(M,g)$ is Riemannian and its {\em injectivity radius} is positive, then other known  ways exist to define a (smooth) notion of squared geodesic distance in a neighbourhood of the diagonal of $M\times M$ (see, e.g.,\cite{Shubin} for the case of a bounded geometry manifold in particular). However, we refer here to  the general case where $(M,g)$ may be  Lorentzian, or  Riemannian with zero injectivity radius.}\\
\end{remark}

An  idea to construct $\Gamma$ and $\sigma$ in $\mathscr{A}$ (see  also the discussion on p.131 of \cite{ONeill}) relies on the insight   that sufficiently small normal convex neighbourhoods are expected to have  intersections which are  normal convex as well. In that case, if $U\cap U'$ is convex and both $p,q \in U$, $p,q\in U'$, then the unique geodesic segment  $ \Gamma_U(t,p,q)  
:= \gamma^{(U)}_{p,q}(t) 
\in U$, $t\in [0,1]$,   joining them in $U$ coincides with the analogue $\Gamma_{U'}(t,p,q) :=
 \gamma^{U'}_{p,q}(t)  \in U'$ joining $p$ and $q$ in  $U'$, since it is the unique geodesic segment joining $p$ and $q$ in $U\cap U'$.  Therefore
$\Gamma_{U}(t,p,q) = \Gamma_{U\cap U'}(t, p,q) = \Gamma_{U'}(t, p,q)$. If a covering $\mathcal{C}$ of $M$ exists made of normal convex open sets such that $U,U' \in \mathcal{C}$ implies that $U\cap U'$ is convex as well, then a jointly smooth assignment of geodesic segments  $\Gamma :[0,1] \times  \mathscr{A} \to M$
joining the arguments $(p,q) \in \mathscr{A}$ (i.e., $\Gamma(0,p,q)=p$ and $\Gamma(1,p,q)=q$)
 is well-defined and smooth on the open neighbourhood $\mathscr{A}$ of $\Delta_M$. It suffices to define 
\begin{align} &\mathscr{A} := \bigcup_{U \in \mathcal{C}}U\times U\label{defA}\\
&\quad \mbox{if $(x,y) \in \mathscr{A}$,}\quad \Gamma(t, x,y) := \Gamma_{U}(t,x,y) \quad \mbox{where $U\in \mathcal{C}$ is such that $x,y\in U$.}\label{defsigma2}
\end{align}
Indeed, if $(x,y) \in \mathscr{A}$, then there must exist $U\in \mathcal{C}$ such that $x,y \in U$.
Next, the right-hand side of (\ref{defsigma}) is well defined, since it does not depend on $U$ if there are other  elements in $\mathcal{C}$ containing $x,y$
 as pointed out above.  $\Gamma$ is also jointly smooth on $\mathscr{A}$ because it is locally jointly smooth. 
In this way, an associated  signed  squared geodesic distance  $\sigma : \mathscr{A} \to \bR$ results  to be well-defined and smooth because it is a  composition of  smooth functions:
\beq \sigma(p,q) := g_p\left( \frac{\partial \Gamma}{\partial t}|_{(0,p,q)},  \frac{\partial \Gamma}{\partial t}|_{(0,p,q)}\right)\:, \quad (p,q) \in \mathscr{A}\:.\label{defSIGMA}\eeq

\begin{definition} If $(M,g)$ is a smooth Riemannian or Lorentzian manifold,  a 
  {\bf strongly convex covering} of $M$ is a
covering ${\cal C}$  of $M$ made of normal convex open sets such that $C\cap C'$ is normal convex if
$C,C'\in {\cal C}$ and $C\cap C' \neq \emptyset$.\\
\end{definition}

Existence of a strongly convex covering ${\cal C}$ is guaranteed when explicitly assuming Hausdorff and  paracompactness hypotheses on $M$\footnote{This idea is sketched in Lemma 10 of Chapter 5 of \cite{ONeill} unfortunately with very few details and without explicitly  referring to the   crucial topological result of Theorem \ref{teoparac}.}.
In fact,  paracompactness possesses  an  important technical feature discovered by  A.H. Stone \cite{Stone} (see also  \cite{Michael}). \\

\begin{theorem} \label{teoparac}  A  topological space $X$ is Hausdorff and paracompact if and only if it is $T_1$ and  every covering ${\cal C}$ of $X$ made of open sets   admits a {\bf $^*$-refinement} of open sets. That is another covering ${\cal C}^*$ of open sets such that, for every $V \in {\cal C}^*$, $$\bigcup \{ V' \in {\cal C}^* \:|\: V' \cap V\neq \emptyset\}\subset U_V\quad \mbox{for some $U_V \in {\cal C}$.}$$
\end{theorem}

(Notice that $V \subset U_V \in {\cal C}$ in particular, so that a  $^*$-refinement is a refinement as well).
This theorem implies the existence of the desired well-behaved  covering $\mathcal{C}$ of normal convex open sets of $(M,g)$ (see also Lemma 10 in Chapter 5 of \cite{ONeill}).\\

\begin{proposition}\label{propintconv} Let $(M,g)$ be  a smooth  (Hausdorff paracompact)  Riemannian or Lorentzian manifold   and ${\cal A}$ a covering of $M$ made of open sets (possibly ${\cal A}:= \{M\}$). Then there exists a covering $\mathcal{C}$ of $M$  sets such that,
\begin{itemize}
\item[(a)] ${\cal C}$ is a refinement of ${\cal A}$ (i.e.,  if $C \in \mathcal{C}$, then   $C \subset U_C\in {\cal A}$) made of  normal convex open sets;
\item[(b)]  if $C, C' \in \mathcal{C}$ and $C\cap C' \neq \emptyset$, then $C \cap C'$ is a normal convex open set. 
\end{itemize}
\end{proposition}

\begin{proof}  Using Theorem \ref{teo01}, consider the covering ${\cal C}_0$ made of all normal convex  neighbourhoods that  are subsets  of the elements of ${\cal A}$.
 Exploiting Theorem  \ref{teoparac}, consider a refinement ${\cal C}^*_0$ of ${\cal C}_0$ satisfying,
 for every $V \in {\cal C}^*_0$, $$\bigcup \{ V' \in {\cal C}^*_0 \:|\: V' \cap V\neq \emptyset\}\subset C_V\quad \mbox{for some $C_V \in {\cal C}_0$.}$$
 The proof concludes by defining ${\cal C}$ as the family of normal convex  neighbourhoods contained within  elements of $ 
{\cal C}^*_0$ so that (a) is in particular true by construction. To prove (b), we start by observing that,   if 
$C,C' \in  {\cal C}$, then $C\subset V$ and $C'\subset V'$ for some $V,V' \in {\cal C}^*_0$;  if furthermore $C\cap C' \neq\emptyset$, we conclude that $V\cap V' \neq \emptyset$ and thus $C\cup C' \subset V\cup V' \subset C_V$.  Property (b)  now  comes  easily using convex normality of $C_V$. First the intersection $C\cap C'$ is open. Next, if $p, q\in C\cap C'$, then the unique geodesic segment $\gamma: [0,1] \to C_V$ joining $p$ and $q$ is also  completely included in $C\cap C'$ since it must simultaneously stay in $C$ and $C'$, they being normal  convex as well. 
As a consequence, if  $p\in C\cap C'$, 
it necessarily  holds  $C\cap C' = \exp_p(V^{(C\cap C')}_p)$ for some star-shaped open neighbourhood  $V^{(C\cap C')}_p := (\exp_p|_{V^{(C)}_p})^{-1}(C\cap C')$
of  the origin of $T_pM$. Notice that  $\exp_p|_{V^{(C\cap C')}_p}: V^{(C\cap C')}_p \to C\cap C'$ is a diffeomorphism because it is the restriction of the  diffeomorphism $\exp_p|_{V^{(C)}_p}: V^{(C)}_p \to C$. In summary, $C\cap C'$ fulfils Definition \ref{defconvex} and the proof is over. \end{proof}

Collecting all results,  we are in a position to state the main theorem of this section,  concening the existence of 
strongly convex coverings in particular,
which also includes a (local) uniqueness statement.\\

\begin{theorem}\label{teo1} Let  $(M,g)$ be a smooth  (Hausdorff paracompact)  Riemannian or Lorentzian manifold, 
and ${\cal C}$ a strongly convex covering of $M$.
Then the following facts hold.
\begin{itemize}
\item[(a)] Defining  the open neighbourhood $\mathscr{A}\supset \Delta_M$ as in (\ref{defA}) with respect to  ${\cal C}$, the assigment of geodesic segments  (\ref{defsigma}) $$\Gamma: [0,1]\times \mathscr{A} \ni (t, p,q) \to \gamma_{p,q}(t) \in M\quad \mbox{where $\gamma_{p,q}(0)=p$ and 
$\gamma_{p,q}(1) = q$,}$$ 
and  the (signed) squared geodesic distance (\ref{defSIGMA}) $$\sigma(p,q)  := g_p(\dot{\gamma}_{p,q}(0), \dot{\gamma}_{p,q}(0)) \quad \mbox{for}\: (p,q) \in \mathscr{A}$$ 
are  well-defined and smooth on $\mathscr{A}$.
\item[(b)] If $\mathscr{A}' \supset \Delta_M$, $\Gamma' : [0,1]\times \mathscr{A}' \to M$, and $\sigma': \mathscr{A}'\to \bR$ is another triple as in (a) but constructed out of another strongly convex covering ${\cal C}'$, then there is an open set  $\mathscr{A}'' \subset M \times M$ such that 
\beq \mathscr{A}\cap \mathscr{A}' \supset \mathscr{A}'' \supset \Delta_M\:, \quad \Gamma|_{[0,1]\times \mathscr{A}''}=\Gamma'|_{[0,1]\times \mathscr{A}''}\:, \quad 
\sigma|_{\mathscr{A}''}= \sigma'|_{\mathscr{A}''}\:.\label{JJ}\eeq
\end{itemize}
\end{theorem}

\begin{proof} (a) 
If $\mathscr{A}$ is  as in  (\ref{defA}), $\Gamma :[0,1] \times  \mathscr{A} \to M$ defined as in  (\ref{defsigma}) 
and $\sigma : \mathscr{A} \to \bR$ defined as in (\ref{defSIGMA})
are well-defined and smooth 
as discussed  in the paragraph before  Eq. (\ref{defA}) and after  Eq. (\ref{defsigma}). (b) 
Define  a new covering ${\cal C}_1$ (a simultaneous refinement of ${\cal C}$ and ${\cal C}'$) made of  the sets $C\cap C'$, for all choices of   $C\in {\cal C}$ and $C'\in {\cal C}'$.
 According to Proposition \ref{propintconv}, define ${\cal C}''$ as a refinement of ${\cal C}_1$  made of normal convex neighborhoods  such that 
$U, U' \in {\cal C}''$ implies that $U\cap U'$ is empty or normal convex and define $\mathscr{A}'' := \cup_{U'' \in \mathcal{C}''} U''\times U''$.  By construction, both $\mathscr{A}'' \subset \mathscr{A}$ and $\mathscr{A}'' \subset \mathscr{A}'$. Moreover, 
if $x,y \in \mathscr{A}''$ then we have both $x,y \in U'' \subset U \in {\cal C}$ and  
$x,y \in U'' \subset U' \in {\cal C}'$, so that  $\Gamma(t,x,y) = \Gamma_{U}(t,x,y) =\Gamma_{U''}(t,x,y) =  \Gamma_{U'}(t,x,y) = \Gamma'(t,x,y)$. The same fact holds for  $\sigma$ and $\sigma'$ in view of their definition (\ref{defSIGMA}) in terms of $\Gamma$ and $\Gamma'$.
\end{proof}
\begin{definition}
{\em A triple $(\mathscr{A}, \Gamma, \sigma)$ as in in (a) of Theorem \ref{teo1} is said to be {\bf subordinated} to ${\cal C}$.}\\
\end{definition}

\begin{remark} {\em Strongly convex coverings are not an
{\em ad-hoc} artifact for the proposal of this work, but a natural and commonly used  technical tool 
in Semi-Riemannian Geometry.  The  existence of this sort of geodesically convex coverings is a straightforward  fact in Riemannian Geometry (see the elementary version of the sketch of proof above when  $h=g$). The 
extension to Lorentzian manifolds  is however  not straightforward.  In addition to the topological approach of Proposition \ref{propintconv},  a purely geometric (in this sense perhaps more natural) proof of existence of strongly convex coverings for a Lorentzian geometry 
can be obtained along the following construction\footnote{The author is grateful to an anonymous referee for the comments contained in this remark,  and to M. S\`anchez  for the idea of the alternative existence proof  of strongly convex coverings sketched here.}.
 (What follows is however valid, with the same proof, when referring to the geodesic flow of a smooth affine connection $\Gamma$ on $M$ which is not the Levi-Civita connection of some metric.)

 Let $(M,g)$ be a smooth  connected (Hausdorff 2nd countable) Lorentzian manifold.
\begin{itemize}

\item[(i)] Let  $h$ be an auxiliary Riemannian metric on $M$ (which exists as a consequence of elementary results in Riemannian geometry \cite{KN}). 
$h$ can be chosen in order that the Riemannian manifold  $(M,h)$ is complete \cite{KO} so that  the $h$-injectivity radius at a given point $p\in M$ is a continuous function of $p$ (see, e.g., Prop. 10.37 in \cite{Lee}).
Consider the atlas of Riemannian normal coordinates $(U_p,\psi_p)$ centered on every $p\in M$ and referred to the Riemannian metric $h$.  

\item[(ii)] Following the classic proof of the Whitehead theorem on the existence of the topological basis of convex normal neighborhoods of $g$ \cite{KN}, one sees that every  $\psi_p$-coordinate ball
$B^h_r(p)$ with center $p$ is $g$-normal convex if the radius $r_p$ is sufficiently small. This would happen if referring to any atlas on $M$, but in the considered case the balls $B^h_r(p)$ are also geodesical balls with respect to $h$ and they are normal and convex with respect to $h$ for sufficiently small $r_p$. As is known, these balls are also metric balls of the natural metric space on $M$ induced by $h$
(the distance $d(p,q)$ is the inf of the $h$-length of the smooth curves joining $p$ and $q$).

\item[(iii)]  Then, one can choose a continuous function $\mu: M\to (0,+\infty)$ -- with
 $2\mu$ smaller than the $h$-radius of injectivity at each point
 $p\in M$ -- such that the ball $B^h_r(p)$  with
 $0<r \leq 2 \mu(p)$ is $g$-normal convex for all $p\in M$.
\end{itemize}
A strong $g$-convex covering is made of  the family of all the balls
 $B^h_{\mu(p)}(p)$, $p\in M$. Indeed,  from elementary properties 
of balls in metric spaces,
the intersection (assumed to be non-empty) of a pair of such 
  balls with centers $p_1$ and $p_2$ is  included in the ball of the center $p_1$
 and (bigger) radius $2\mu(p)$. This ball  is $g$-convex by construction.}
\end{remark}
\section{An issue with the global Hadamard condition and Hadamard parametrix}\label{secHcond}
Before addressing another  issue still related to the well-definedness of $\sigma$ and associated structures,  we  summarize the relevant notions introduced  by  the milestone paper \cite{KW} where, for the first time, a rigorous definition of a Hadamard state was proposed and used by B.S. Kay and R.M. Wald. The definition was used  in \cite{KW}  (relying on previous work as \cite{FSW} and \cite{Kay}) to establish some important uniqueness results of QFT on a spacetime  equipped with a bifurcate Killing horizon  related to the KMS states  of a real Klein-Gordon scalar field  with the Hawking temperature.  However, the  definition of Hadamard state discussed therein  applies to every (four-dimensional) globally hyperbolic spacetime. The notion of Hadamard states in Kay-Wald's approach relies upon the notion of {\em Hadamard parametrix}. The Hadamard condition on states
can be nowadays formulated without a Hadamard parametrix using microlocal techniques as we shall briefly discuss later. It is however worth stressing that  the notion of Hadamard parametrix remains a crucial technical  tool for  the construction of other important mathematical objects in QFT as the {\em Wick powers} in the {\em locally covariant formulation} (see in particular \cite{HW,KM}).

\subsection{Hadamard states according to \cite{KW}}  

If $(M,g)$ is a time-oriented smooth  spacetime and  $x,y \in M$,
\beq J(x,y):= (J^{-}(x) \cap J^+(y)) \cup (J^{-}(y) \cap J^+(x))\:,\label{JJ23}\eeq
where $J^\pm(S)$ are defined as in \cite{ONeill}. We say that $x,y$ are  {\bf causally related} in $(M,g)$  if $J(x,y)\neq \emptyset$.
We henceforth assume that $(M,g)$ is four dimensional and globally hyperbolic.\\

\begin{remark} \label{remJ} {\em  If $x,y\in M$ are causally related in a globally-hyperbolic spacetime $(M,g)$, then  there is a causal geodesic segment joining them in view of Proposition 19 in Chapter 14 of \cite{ONeill}. This fact has a crucial consequence.
If $x,y$ are causally related 
and both $U \supset J(x,y)$, $U' \supset J(x,y)$ for  convex normal  neighbourhoods $U,U'$,
then   $\sigma_U(x,y) =\sigma_{U'}(x,y)$.   Indeed, the unique geodesic segments parametrized on $[0,1]$ connecting $x$ and $y$ respectively in $U$ and $U'$ must  belong to $J(x,y) \subset U\cap U'$ and thus they must coincide.  
This  fact is throughout   exploited in \cite{KW}  and provides a well-defined notion of signed squared geodesic distance $\sigma(x,y)$ on the subset of $M\times M$
$$\mathscr{Z}_M := \{ (x,y)\in M\times M \:|\:  x,y \: \mbox{causally related}\:, \:   J(x,y)\subset U, \: U\:  \mbox{normal convex neighbourhood}\}\:.$$}
\end{remark}

The definition of Hadamard state according \cite{KW} passes through the following four steps.
\begin{itemize}
\item[{\bf (H1)}]
The so-called (global) {\bf Hadamard parametrix} is defined in \cite{KW}, for every natural $n$ and $\epsilon>0$, as
\beq\label{Gor}
G^{T,n}_\epsilon(x,y) := \frac{1}{(2\pi)^2}\left[\frac{\Delta(x,y)^{1/2}}{\sigma(x,y) + 2i\epsilon t(x,y) + \epsilon^2}
+v_n(x,y) \ln (\sigma(x,y) +2i\epsilon t(x,y) +\epsilon^2) \right], \quad (x,y) \in \mathscr{O}\:.
\eeq
Above  $\mathscr{O}\supset \mathscr{Z}_M$ is an open set {\em supposed to exist} where $\sigma$ and $G^{T,n}_\epsilon$ are  well defined, $t(x,y) := T(x)-T(y)$, where the smooth function  $T: M \to \bR$ is a {\em Cauchy temporal function}\footnote{I.e., $dT \neq 0$  is  everywhere past-directed and $T^{-1}(r)$ is a smooth spacelike Cauchy surface for every $r\in \bR$.} increasing towards the future, the branch cut of the logarithm is taken along the negative real axis, and the function $\Delta(x,y)$ and $v_n(x,y)$ are known and  defined in terms  of $\sigma(x,y)$ and known recursion  integrals along the geodesic segment $\gamma_{xy}$ connecting $x$ and $y$ (see, e.g., Appendix A of \cite{Stress} and \cite{HM}). \\

\begin{remark} \label{remvs}{\em If $\sigma(x,y)$ and the geodesic segment $\gamma_{xy}$ connecting $x$ and $y$  are well defined  in some  neighborhood, then $\Delta(x,y)$ and $v_n(x,y)$ are completely determined in that neighbourhood.
This happens in particular for $x,y \in U$ with $U$ normal convex neighbourhood.}\\
\end{remark}

\item[{\bf (H2)}] Following  \cite{KW}, given a globally hyperbolic spacetime $(M,g)$ with a time orientation and a smooth spacelike Cauchy surface $\Sigma$, a {\bf normal neighbourhood} $N$ of $\Sigma$ is an open set including $\Sigma$ and such that
\begin{itemize}
\item[(a)] $(N,g|_N)$ is a globally hyperbolic spacetime and $\Sigma$ is a Cauchy surface of it;
\item[(b)] $(x,y) \in N\times N$ are causally related in $(M,g)$ iff  $(x,y) \in \mathscr{Z}_M$. 
\end{itemize}
Lemma 2.2 of \cite{KW} proves the existence of a normal neighbourhood of any given Cauchy surface $\Sigma$. 

\item[{\bf (H3)}] Consider an open set $\mathscr{O}' \subset N\times N$ which includes $\mathscr{Z}_M \cap (N \times N)$ (i.e., the set of causally related pairs $(x,y) \in N\times N$) and such that its closure in $N\times N$ satisfies $\overline{\mathscr{O}'}^{N\times N}\subset \mathscr{O}$. Finally,
$\chi: N\times N \to \bR$ is a smooth function such that $\chi(x,y)=1$ for $(x,y) \in \overline{\mathscr{O}'}^{N\times N}$ and $\chi(x,y)=0$ 
for $(x,y) \not\in \mathscr{O}\cap (N\times N)$.
\item[{\bf (H4)}] 
With $\mathscr{O}$, $N$, $T$, $\chi$ as above, we can state  the definition of Hadamard state. 
 \begin{definition} \label{defHSC} {\em An algebraic  state  $\omega$ on the (Weyl $C^*$ or $^*$) algebra of a real scalar Klein-Gordon field on $(M,g)$ is said to be {\bf globally Hadamard} according to \cite{KW} if the associated 
{\em two-point function}, i.e., a certain bilinear map \cite{KW}  $\Lambda_\omega : C_0^\infty(M) \times C_0^\infty(M)\to \bC$,  satisfies the following requirement 
\beq
\Lambda_\omega(F_1,F_2) = \lim_{\epsilon \to 0^+} \int_{N\times N} \Lambda^{T,n}_\epsilon(x,y) F_1(x) F_2(y) d\mu_g(x) d\mu_g(y)\:, \quad 
\forall F_1,F_2 \in C_0^{\infty}(N)\:,
\eeq
where $\mu_g$ is the natural measure induced by $g$ on $M$ and
\beq\label{coc}
 \Lambda^{T,n}_\epsilon(x,y) = \chi(x,y) G^{T,n}_\epsilon(x,y)  + H^n(x,y)\:,
\eeq
for every natural $n$ and some associated functions $H^n\in C^n(N\times N)$.}
\end{definition}
\end{itemize}

\begin{remark}
{\em In \cite{KW}, it is proved  that Definition \ref{defHSC}  is independent of $\mathscr{O}, N, \chi, \Sigma$. Yet, that independence proof assumes at various steps that $G^{T,n}_\epsilon(x,y)$ is well defined, not only on $\mathscr{Z}_M$, but also on $\mathscr{O}$  (and $\mathscr{O}'$).  In particular,   $\sigma(x,y)$ is  expected to have the standard behaviour in $\mathscr{O}$: $\sigma(x,y) >0$ if $x\neq y$ are not causally related. More precisely, $\sigma(x,y)$ is supposed to  take the standard form $\sigma(x,x')=-(y^0(x'))^2 + \sum_{\alpha=1}^3 (y^\alpha(x'))^2$ in {\em Riemannian normal coordinates} $y^0,y^1,y^2,y^3$ centered at one of the entries (here $x$) also for non-causally related arguments.}\\
\end{remark}

Definition \ref{defHSC}  was later proved to be equivalent
 to a  certain {\em microlocal version} by a  famous  paper by  M. Radzikowski \cite{Rad}, when assuming  the   requirement $\Lambda_\omega \in {\cal D}'(M\times M)$ (see (2) in Theorem \ref{teoRAD} below). This second analytic version (extended to $n$-dimensional spacetimes with $n\geq 2$) is the one usually nowadays adopted in perturbative aQFT, also  including  cosmological applications, starting form semiclassical versions of the Einstein equations (see \cite{MPS} for a recent application). 
 The Hadamard parametrix plays a special role in the definition of locally-covariant Wick powers \cite{HW,KM} and in the study of   quantum energy inequalities \cite{FS08}.  Kay-Wald's version  of the Hadamard condition    has been used  by R. Verch to prove physically important properties 
of Hadamard states at algebraic level, like {\em local quasi equivalence} and {\em local definiteness} \cite{Verch}. 
 Using Kay-Wald's definition, Sahlmann and Verch \cite{SV}  extended the formalism to vector-valued quantum fields in a globally hyperbolic  spacetime of dimension $n\geq 2$. There,  also  the equivalent microlocal formulation has been discussed and an extension  of the {\em theorem of propagation of Hadamard singularity} has been established in the fashion of the original formulation \cite{FSW} of that property of Hadamard states.  We recommend \cite{book} for a recent account on the wide spectrum of applications of Hadamard states (a pedagogical  introduction 
 to quasifree Hadamard states and their relevance in aQFT
takes place
in  \cite{KM} therein). 

The specific use of the Hadamard condition in the study of Hawking radiation can  be traced back to \cite{FH}, already before that the precise form of the Hadamard parametrix was stated in \cite{KW}.
Though the microlocal version has been recently employed   in  applications to aQFT in black-hole background \cite{DMP,Sanders},
the originary \cite{KW} version of the Hadamard condition has continued  to  play a crucial role  to discuss  the Hawking  radiation \cite{Wald},  also in terms of a tunneling process   \cite{MP,CMP} (actually, those works only concern a local version of the Hadamard condition). See in particular the  recent interesting  work \cite{KNV} on  the Hawking radiation (and partially on the black-hole entropy) for a collapsing black-hole spherically-symmetric spacetime, where the global Hadamard condition has been used.

An interesting  global   definition of Hadamard state has been recently discussed in \cite{CDD}  in terms of pseudo differential operators and a different, but related, notion of {\em global parametrix} for globally hyperbolic spacetimes with compact Cauchy surfaces.

  \subsection{A gap in  the definition of $G^{T,n}_\epsilon$ and a proposal of solution}
The parametrix $G^{T,n}_\epsilon$ is evidently well defined on $\mathscr{Z}_M$, but there is no guarantee that it is also well defined on some open neighbourhood $\mathscr{O} \supset \mathscr{Z}_M$. Indeed,  the open set $\mathscr{O}$ must also contain pairs $(p,q)$ which are {\em not} causally related and each such pair  may be connected by many geodesic segments because Remark \ref{remJ} does not apply.  At this juncture, there is no explicit prescription to smoothly  choose a unique geodesic segment for every  such pair $(p,q)$ in order to have a well defined $\sigma(p,q)$, which, e.g.,  satisfies $\sigma(x,y)>0$ when  $x\neq y$ are not causally related. The problem also arises in the definition of $\Delta(p,q)$ and $v^n(p,q)$ as they are computed  using a geodesic segment joining $p$ and $q$ as said above. 

Instead of  attacking the problem directly by trying to establish the existence of a neighbourhood $\mathscr{O} \supset \mathscr{Z}_M$
where $\sigma$ and $G_\epsilon^{T,n}$ are well defined, we adopt a different strategy to circumvent the gap by  employing the achievements of Sect \ref{SectS}. The strategy relies on minimal modifications of original Kay-Wald's machinery. {\em For this reason,   in author's view, all important results established over the years  that  rely on Definition
\ref{defHSC}  (some of them quoted above) are  correct.}

Given a four-dimensional  globally hyperbolic spacetime $(M,g)$ with a time orientation,
choose a strong convex covering ${\cal C}$ of $M$, define the  triple $(\mathscr{A}, \Gamma, \sigma)$ subordinated to ${\cal C}$ as in Theorem \ref{teo1} and the set
$$\mathscr{Z}^{\cal C}_M := \{ (x,y)\in M\times M \:|\:  x,y \: \mbox{causally related}\:, \:   J(x,y)\subset U \in {\cal C}\}\:.$$
Notice that $\mathscr{A}$ is an open neighborhood of $ \mathscr{Z}^{\cal C}_M$  by construction.
\begin{itemize}
\item[{\bf (H1)'}]
Define a (global)  {\bf Hadamard parametrix subordinated to} ${\cal C}$, for every natural $n$ and $\epsilon>0$, as
\beq
G^{T,n,{\cal C}}_\epsilon(x,y) := \frac{1}{(2\pi)^2}\left[\frac{\Delta(x,y)^{1/2}}{\sigma(x,y) + 2i\epsilon t(x,y) + \epsilon^2}
+v^n(x,y) \ln (\sigma(x,y) +2i\epsilon t(x,y) +\epsilon^2) \right], \quad (x,y) \in \mathscr{A}\:.
\eeq
Above, $t(x,y) := T(x)-T(y)$, where $T: M \to \bR$ is global  smooth time function increasing towards the future, the branch cut of the logarithm is taken along the negative real axis, and the functions, $\sigma$, $\Delta$ and $v_n$ are  the ones constructed out of  $(\mathscr{A},\Gamma, \sigma)$ starting from ${\cal C}$.

\item[{\bf (H2)'}] Given a smooth spacelike Cauchy surface $\Sigma$ of $(M,g)$ (with dimension $\geq 2$), a {\bf normal neighbourhood} $N_{\cal C}$ of $\Sigma$ {\bf subordinated to} ${\cal C}$  is an open set including $\Sigma$ and such that
\begin{itemize}
\item[(a)] $(N_{\cal C} ,g|_{N_{\cal C}})$ is a globally hyperbolic spacetime and $\Sigma$ is a Cauchy surface of it;
\item[(b)] $(x,y) \in N_{\cal C}\times N_{\cal C} $ are causally related in $(M,g)$ iff $(x,y) \in \mathscr{Z}^{\cal C}_M$. 
\end{itemize}

\begin{lemma}
Given a strong convex covering of $M$, every smooth spacelike Cauchy surface of $(M,g)$ admits a normal neighbourhood subordinated to ${\cal C}$.
\end{lemma}

\begin{proof} Use the same proof as the one of 
Lemma 2.2 of \cite{KW}  with the only difference that all the  used  normal convex neighbourhoods must be taken in ${\cal C}$.
\end{proof}

\item[{\bf (H3)'}] Consider an open set $\mathscr{A}' \subset N_{\cal C}\times N_{\cal C}$ which includes $\mathscr{Z}^{\cal C}_M \cap (N_{\cal C} \times N_{\cal C})$ (i.e., the set of causally related pairs $(x,y) \in N_{\cal C}\times N_{\cal C}$) and such that its closure in $N_{\cal C}\times N_{\cal C}$ satisfies $\overline{\mathscr{A}'}^{N_{\cal C}\times N_{\cal C}}\subset \mathscr{A} \cap (N_{\cal C}\times N_{\cal C})$.

\begin{remark} {\em $\mathscr{A}'$ does exist because $\mathscr{Z}^{\cal C}_M \cap (N_{\cal C} \times N_{\cal C})$ is closed in $N_{\cal C} \times N_{\cal C}$ (with the relative topology) and it is included in the open set $\mathscr{A} \cap (N_{\cal C} \times N_{\cal C})$.
(The set $\mathscr{Z}^{\cal C}_M \cap (N_{\cal C} \times N_{\cal C})$ of causally related points in $N_{\cal C}$  is closed in $N_{\cal C} \times N_{\cal C}$ because  $N_{\cal C}$ is globally hyperbolic and  Lemma 22 in Chapter 14 of \cite{ONeill} is valid\footnote{The proof appearing in Lemma 3.3 in \cite{Rad} of this fact  seems to be wrong or incomplete, because $\sigma$ is not necessarily defined in the target points of the considered sequences, though it grasps the  correct insight, with an appropriate use of  the {\em time separation function} $\tau$ in place of $\sigma$ as in  Chapter 14 of \cite{ONeill}.}.)}
\end{remark}

 Finally, taking advantage of the smooth Urysohn  lemma, choose a smooth function
$\chi: N_{\cal C}\times N_{\cal C} \to [0,1]$  such that $\chi(x,y)=1$ for $(x,y) \in \overline{\mathscr{A}'}^{N_{\cal C}\times N_{\cal C}}$ and $\chi(x,y)=0$ 
for $(x,y) \not\in \mathscr{A}\cap (N_{\cal C}\times N_{\cal C})$.

\item[{\bf (H4)'}] With ${\cal C}$, $N_{\cal C}$,$T$, $\chi$ as above, we can give the definition of Hadamard state. 
\begin{definition} \label{defNhad} {\em An algebraic  state  $\omega$ on the (Weyl $C^*$ or $^*$) algebra of a real scalar Klein-Gordon field on $(M,g)$ is said to be {\bf globally Hadamard}   if the associated 
{\em two-point function} $\Lambda_\omega : C_0^\infty(M) \times C_0^\infty(M)\to \bC$,  satisfies the following requirement 
\beq\label{Hc}
\Lambda_\omega(F_1,F_2) = \lim_{\epsilon \to 0^+} \int_{N_{\cal C}\times N_{\cal C}} \Lambda^{T,n,{\cal C}}_\epsilon(x,y) F_1(x) F_2(y) d\mu_g(x) d\mu_g(y)\:, \quad 
\forall F_1,F_2 \in C_0^{\infty}(N_{\cal C})\:,
\eeq
where $\mu_g$ is the natural measure induced by $g$ on $M$ and
\beq\label{tred}
 \Lambda^{T,n,{\cal C}}_\epsilon(x,y) = \chi(x,y) G^{T,n,{\cal C}}_\epsilon(x,y)  + H^n(x,y)\:,
\eeq
for every natural $n$ and some associated functions $H^n\in C^n(N_{\cal C}\times N_{\cal C})$.}
\end{definition}
\end{itemize}

\begin{remark}\label{remFIN} {\em An identity is of utmost  physical interest: two parametrices subordinated to different strong convex coverings are however identical 
and also coincide with $G^{T,n}_\epsilon(x,y)$ in (\ref{Gor})
$$G^{T,n,{\cal C}}_\epsilon(x,y)=G^{T,n,{\cal C}'}_\epsilon(x,y)=G^{T,n}_\epsilon(x,y)$$ 
{\em when evaluated 
on causally related points} $(x,y)\in N_{\cal C} \cap N_{{\cal C}'}$. In fact, in the said hypothesis, it simultaneously holds
$(x,y) \in J(x,y) \subset C\in {\cal C}$ and $(x,y) \in J(x,y) \subset C'\in {\cal C}'$ and thus, according to Remark \ref{remJ}, the geodesic segments joining $x$ and $y$ in $C$ and $C'$, respectively, coincide.
Finally the parametrices coincide as well in view of Remark \ref{remvs}. 
What happens to $G^{T,n,{\cal C}}_\epsilon(x,y)$ for non-causally related points is physically irrelevant  and it  permits  an arbitrary choice of the function $\chi$ appearing in  $\chi(x,y)G^{T,n,{\cal C}}_\epsilon(x,y)$. A change of the function $\chi$  can be reabsorbed  in a change  of the functions $H^n$. That  is a consequence of the fact that, 
for $x\neq y$ non-causally-related, $\sigma(x,y)>0$  and no singularity (for $\epsilon \to 0^+$) shows up in the parametrix $G^{T,n,{\cal C}}_\epsilon(x,y)$. In other words, the parametrix viewed as a distribution is actually a smooth function for non-causally-related arguments.
All that was discussed and  clearly emphasized in \cite{KW}  referring to the parametrix $G^{T,n}_\epsilon$. Unfortunately these properties of 
$G^{T,n}_\epsilon$ rely also on a good behavior of $\sigma$ in the whole open neighborhood $\mathscr{O}$ (and $\mathscr{O}'$) which is not proved to exist.}
\end{remark}
\subsection{Independence of the choices of  ${\cal C}, N_{\cal C}, T, \chi$ and nice interplay with the microlocal formulation}
What remains to be demonstrated  is that the given definition of Hadamard state does not depend on the choice of ${\cal C},  N_{\cal C}, T, \chi$ and it corresponds to the microlocal formulation \cite{Rad}.
\cite{Rad}  aimed to establish that a  state of a real  Klein-Gordon field
 in a globally hyperbolic  spacetime $(M,g)$  is Hadamard in the sense of \cite{KW} if and only if  it satisfies the 
{\em microlocal spectral condition} (\ref{SC}) below. (Actually  it was done  when also assuming the fair hypothesis that  the two-point function $\Lambda$ is a distribution of ${\cal D}'(M\times M)$.)  As a matter of fact, this result gave rise to an alternative  definition of Hadamard state.

The presentation of the Hadamard condition in the original sense of \cite{KW} in \cite{Rad,Rad2} is affected by the issue pointed out above  (in the proof of Lemma 3.1  in \cite{Rad} in particular) since  \cite{Rad} includes a faithful summary of relevant ideas and notions appearing in  \cite{KW}.

We argue that the statement of  Theorem 5.1 in \cite{Rad} which establishes the equivalence of the two formulations {\em is however valid when assuming our  definition of Hadamard state according to (H1)'-(H4)'}.  Let us re-state here part of  Radzikowski's equivalence theorem (excerpt  from Theorem 5.1 \cite{Rad} with notations adapted to our paper).\\

\begin{theorem} \label{teoRAD} Let $(M,g)$ be  a smooth, time oriented, four-dimensional globally hyperbolic spacetime and $\Lambda \in {\cal D}'(M\times M)$, define the  Klein-Gordon operator $P:= -\Box + m^2: C^\infty(M) \to C^\infty(M)$ for some real valued $m^2\in C^\infty(M)$. Choose ${\cal C}, T, \chi, N_{\cal C}$ as above. Then the following conditions are equivalent.
\begin{itemize}
\item[(1)]  $\Lambda$ 
satisfies what follows. \begin{itemize}\item[(a)]  The global Hadamard condition in Definition \ref{defNhad} (referring to the given choice of ${\cal C}, T, \chi, N_{\cal C}$), \item[(b)]  its antisymmetric part is $\frac{i}{2}(\Delta_A - \Delta_R )$  (where $\Delta_{A/R}: C_0^\infty(M) \to C^\infty(M)$ are the advanced/retarded Green operators of $P$), \item[(c)]  $\Lambda(PF\otimes F')= \Lambda(F\otimes PF')=0$ for all real-valued  $F,F' \in C_0^\infty(M)$.
\end{itemize}
\item[(2)] $\Lambda$ satisfies what follows.
\begin{itemize}
\item[(a)']  The {\bf microlocal spectral condition}
\beq WF(\Lambda) = \{ ((x_1,k_1), (x_2,k_2) \in T^*M \setminus {\bf 0} \times T^*M \setminus {\bf 0} \:|\:  (x_1,k_1), (x_2,-k_2)\:, k_1 \rhd 0 \}\:, \label{SC} \eeq
\item[(b)]  its antisymmetric part is $\frac{i}{2}(\Delta_A - \Delta_R )$, \item[(c)]  $\Lambda(PF\otimes F')= \Lambda(F\otimes PF')=0$
for all real-valued  $F,F' \in C_0^\infty(M)$.
\end{itemize}
\end{itemize}
((b) and (c) are in particular  valid if $\Lambda= \Lambda_\omega$ for an algebraic state $\omega$ on the (Weyl $C^*$ or $^*$) algebra of a real Klein Gordon quantum field.)\\
Equivalence of (1) and (2)  is still valid with the following changes. (b) and (c) in (1) are true  mod $C^\infty$,  (b) in (2) is  true  mod $C^\infty$, and (c) is omitted from (2).\\
\end{theorem}

\noindent {\em Sketch of Proof}. The proof of  of Theorem 5.1 \cite{Rad}
uses both microlocal analysis arguments and some results from \cite{KW}.
Concerning definitions and facts established in \cite{KW}, it is assumed
that (i) the parametrix $G^{T,n}_\epsilon(x,y)$ has the known structure in terms of $\sigma$, $\Delta$, $v^n$ in a covering of normal convex neighbourhoods as stated in \cite{KW},   (ii) the Hadamard expansion is well-behaved on the  open neighborhood $\mathscr{O}'$ of the causally related points in $N\times N$, where $\sigma(x,y)>0$ for $x\neq y$ which are not causally related (more precisely it takes the standard form 
$\sigma(x,x')=-(y^0(x'))^2 + \sum_{\alpha=1}^3 (y^\alpha(x'))^2$ in  Riemannian normal coordinates $y^0,y^1,y^2,y^3$ centered at $x$), and  (iii)  the definition of Hadamard state according to \cite{KW} is independent from the choice of the Cauchy temporal function  $T$. 
 This  proof of independence appears  in Appendix B of \cite{KW} and it can be recast without changes for our  definition of Hadamard state based on the parametrix  $G^{T,n, {\cal C}}_\epsilon(x,y)$ and a normal neighborhood $N_{\cal C}$.
In summary, replacing $\mathscr{O}'$ for $\mathscr{A}'$, using the fact that $G^{T,n, {\cal C}}_\epsilon(x,y)$
has the same local structure as $G^{T,n}_\epsilon(x,y)$ in terms of $\sigma$ and the Hadamard expansion coefficients  and  {\em is} well-behaved on $\mathscr{A}$,  and exploiting independence of the definition from  the choice of $T$,
 the proof  of Theorem 5.1 \cite{Rad} is valid as it stands 
for Definition \ref{defNhad} of Hadamard state\footnote{As is known, the proof of Theorem 5.1 in \cite{Rad} has a gap.
It is  the content of the three lines immediately before the proof of (ii) 3 $\Rightarrow$ 2 on p.547.
This gap was  closed in several independent works, in particular (but not only)  \cite{SV} and \cite{KM}. In  the latter, only the microlocal analysis approach was exploited and thus without  relation with the issue with  \cite{KW} definition of Hadamard state. See Remark 23 in  \cite{KM} for a summary on this subject.}. \hfill $\Box$\\

\begin{corollary} The definition of Hadamard state in Definition \ref{defNhad} (based on (H1)'-(H4)' and assuming (b) and (c) of Theorem \ref{teoRAD} for $\Lambda_\omega$)  does not depend on the choices of  ${\cal C},  N_{\cal C}, T, \chi$ if $\Lambda_\omega \in {\cal D}'(M\times M)$.
\end{corollary}
\begin{proof}
Item (2) above does not depend on the choices of  ${\cal C}, N_{\cal C}, \chi$. (Independence from the choice of $T$ was independently established in the proof of  Theorem \ref{teoRAD} using the same proof as in \cite{KW}.)
\end{proof}

To conclude we observe that, following \cite{Rad2},  an algebraic  state $\omega$  on a Klein-Gordon quantum field on a spacetime $(M,g)$  is said to be {\bf locally Hadamard} if there is a (normal convex) neighborhood  $U$ of every point where, for every natural $n$, the two-point function of the state $\Lambda_\omega$ can be decomposed as in (\ref{coc}) (i.e., (\ref{tred})) with $\Lambda_\omega$ in place of $\Lambda^{T,n}$,  for $\chi=1$ and $H^n \in C^n(U\times U)$. It is possible to prove that a state $\omega$ such that $\Lambda_\omega \in {\cal D}'(M\times M)$ satisfies  (b) in Theorem \ref{teoRAD}  is locally Hadamard in a four-dimensional globally hyperbolic spacetime  if and only if it is globally Hadamard. It was established in Theorem 9.2 in  \cite{Rad2} using only  the microlocal definition (i.e., (2) in Theorem \ref{teoRAD}) of Hadamard state  and thus that result is valid also with our definition of (global) Hadamard state.

\section*{Acknowledgments}  I am  grateful to Franco Cardin, Claudio Dappiaggi,  Nicol\`o Drago,  Nicola Pinamonti, and  Miguel S\`anchez   for various  discussions over the years on  subjects somehow related to the content of this work. I thank Chris Fewster  and Miguel S\`anchez   for several technical discussions and suggestions about this paper.  I am grateful to the two anonymous referees for a number of suggestions and remarks of various nature which really helped  me improve the presentation of this work.
  I am finally grateful to C.  Fewster, F. Kurpicz, N. Pinamonti, and R. Verch who, directly or indirectly,  encouraged the author to write down this  quite technical note.


\begin{thebibliography}{999}


\bibitem[BFK96]{BF96} R.Brunetti, K. Fredenhagen, M., K\"ohler, {\em The microlocal spectrum condition and Wick
polynomials of free fields on curved space-times.}  Commun. Math. Phys. 180, 633-652 (1996)

\bibitem[BF00]{BF00} R. Brunetti, K. Fredenhagen, {\em Microlocal analysis and interacting quantum field theories:
renormalization on physical backgrounds.} Commun. Math. Phys. 208, 623-661 (2000)



\bibitem[BDFY15]{book}   R. Brunetti, C. Dappiaggi, K. Fredenhagen, and J.Yngvason Editors,   {\em Advances in Algebraic Quantum Field Theory}. Springer (2015) 

\bibitem[CDD20]{CDD} M. Capoferri, C. Dappiaggi, N. Drago,  {\em Global wave parametrices on globally hyperbolic spacetimes}, 
J. Mathemarical  Analysis and  Applications, 490, 2, 124316,  (2020).


\bibitem[CM04]{MC} F. Cardin and  A. Marigonda, {\em Global world functions}, J. Geom. Sym. Phys. 2: 1-17 (2004). 


\bibitem[CMP14]{CMP}  G. Collini,  V.Moretti  and N. Pinamonti, {\em Tunnelling black-hole radiation with $\phi^3$ self-interaction: one-loop computation for Rindler Killing horizons}
Lett. Math. Phys. 104 (2014) 217-232




\bibitem[DFP08]{DFP08} C. Dappiaggi, K. Fredenhagen, N. Pinamonti, {\em Stable cosmological models driven by a free
quantum scalar field}. Phys. Rev. D 77, 104015 (2008)



\bibitem[DMP09]{DMP2} C. Dappiaggi, V. Moretti, N. Pinamonti, {\em Distinguished quantum states in a class of 
cosmological spacetimes and their Hadamard property}. J. Math. Phys. 50, 062304 (2009).

\bibitem[DMP11]{DMP}  C.Dappiaggi,  V. Moretti,  N. Pinamonti, {\em Rigorous  construction  and  Hadamard  propertyof the Unruh state in Schwarzschild spacetime}.  Adv. Theor. Math. Phys. 15, 355-447 (2011)

\bibitem[FS08]{FS08}   C.J. Fewster, C.J. Smith, {\em Absolute quantum energy inequalities in curved spacetime.} Ann.
Henri Poincar\'e 9, 425-455 (2008)


\bibitem[FV03]{FV03} C.J. Fewster, R. Verch, {\em Stability of quantum systems at three scales: passivity, quantum
weak energy inequalities and the microlocal spectrum condition.} Commun. Math. Phys. 240,
329-375 (2003)

\bibitem[FV13]{FV13}  C.J. Fewster, R. Verch, {\em The necessity of the Hadamard condition.} Class. Quant. Grav. 30,
235027 (2013)



\bibitem[FH90]{FH}  K. Fredenhagen and R. Haag, {\em On the derivation of Hawking radiation associated with the formation of a black hole}, Commun. Math. Phys. 127, 273 (1990)

\bibitem[FSW78]{FSW} S. A. Fulling, M. Sweeny and R. M. Wald, {\em Singularity structure of the two-point function in quantum field theory in curved spacetime},
Commun. Math. Phys. 63 (1978) 257-264

\bibitem[HM12]{HM} T.P. Hack and  V. Moretti,  {\em On the stress-energy tensor of quantum fields in curved spacetimes-comparison of different regularization schemes and symmetry of the Hadamard/Seeley-DeWitt coefficients}, J. Physics A: Mathematical and Theoretical 45 (37), 374019

\bibitem[HW01]{HW} S. Hollands and R.M.  Wald, {\em Local Wick polynomials and time ordered products of quantum fields in curved spacetime}. Commun. Math. Phys. 223, 289-326 (2001). arXiv:gr-qc/0103074

\bibitem[HW02]{HW02} S. Hollands, S., R.M. Wald, {\em Existence of local covariant time ordered products of quantum
fields in curved spacetime.} Commun. Math. Phys. 231, 309-345 (2002)


\bibitem[KM15]{KM}  I. Khavkine and V. Moretti, {\em Algebraic QFT in Curved Spacetime and quasifree Hadamard states: an introduction}, Chapter 5  of \cite{book} 

\bibitem[KO61]{KO} K. Nomizu and H. Ozeki, {\em The Existence of Complete Riemannian Metrics}. Proceedings of the American Mathematical Society
Vol. 12, No. 6 (Dec., 1961), pp. 889-891

\bibitem[Ka85]{Kay} B.S. Kay, {\em A uniqueness result for quasi-free KMS states,} HeIv. Phys. Acta 58 (1985) 1017-1029.

\bibitem[KN96]{KN} K. Kobayashi and S. Nomizu: {\em Foundations of Differential Geometry}. Vol I, (1996)

\bibitem[KPV21]{KNV} F. Kurpicz, N. Pinamonti, R. Verch, {\em Temperature and entropy-area relation of quantum matter near spherically symmetric outer trapping horizons}, Lett. Math. Phys. (2021) in print. 
arXiv preprint arXiv:2102.11547

\bibitem[KW91]{KW} B.S. Kay,  R.M. Wald, {\em Theorems on the uniqueness and thermal properties of stationary, nonsingular,
quasifree states on spacetimes with a bifurcate Killing horizon}. Phys. Rep. 207(2), 49-136 (1991)

\bibitem[Lee18]{Lee} J. Lee, {\em Introduction to Riemannian Manifolds}. Springer (2018)


\bibitem[Mi59]{Michael} E. Michael, {\em Yet Another Note on Paracompact Spaces},
Proceedings of the American Mathematical Society , Apr., 1959, Vol. 10, No. 2
(Apr., 1959), pp. 309-314.

\bibitem[Mi19]{Min}
E Minguzzi, {\em Lorentzian causality theory},
Living reviews in relativity 22 (1), 1-202 (2019)

\bibitem[Mo03]{Stress} V. Moretti, {\em  Comments on the Stress-Energy Tensor Operator
in Curved Spacetime}, Commun. Math. Phys. 232, 189-221 (2003)

\bibitem[Mo08]{Mo08}  V. Moretti, {\em Quantum out-states holographically induced by asymptotic flatness: invariance
under spacetime symmetries, energy positivity and Hadamard property.} Commun. Math. Phys.
279, 31-75 (2008)


\bibitem[MP12]{MP}  V. Moretti,  N. Pinamonti, {\em State independence for tunneling processes through black hole horizons},
Commun. Math. Phys. 309 (2012) 295-311 


\bibitem[MPS21]{MPS} P. Meda, N. Pinamonti, D. Siemssen, {\em Existence and uniqueness of solutions of the semiclassical Einstein equation in cosmological models}.
Annales Henri Poincar\'e, (online first, Published: 28 June 2021), 1-51 (2021)


\bibitem[ON83]{ONeill}
B. O'Neill, {\em Semi-Riemannian Geometry With Applications to Relativity}. Academic
Press (1983)


\bibitem[Rad96a]{Rad} M. J. Radzikowski,  {\em Microlocal approach to the Hadamard condition in quantum field theory on curved
space-time}, Commun. Math. Phys. 179, 529-553 (1996)

\bibitem[Rad96b]{Rad2}  M. J. Radzikowski (with an Appendix by Rainer Verch)
{\em A Local-to-Global Singularity Theorem for Quantum Field Theory on Curved Space-Time}, Commun. Math. Phys. 180, 1-22 (1996)

\bibitem[Sa15]{Sanders} K. Sanders,  {\em On  the  construction  of  Hartle-Hawking-Israel  states  across  a  static  bifurcateKilling  horizon.}  Lett.    Math.  Phys. 105,  575-640  (2015).


\bibitem[Sh91]{Shubin} M.A. Shubin, {\em Spectral theory of elliptic operators on non-compact manifolds, dans M\'ethodes 
semi-classiques} Volume 1- \'Ecole d'\'Et\'e (Nantes, juin 1991), Ast\'erisque, no. 207 (1992), 74 p. 

\bibitem[St49]{Stone} A. H. Stone, {\em Paracompactness and product spaces}, Bull. Amer. Math. Soc. vol.
54 (1948) pp. 977-982 


\bibitem[SV01]{SV} H. Sahlmann and R. Verch, {\em Microlocal spectrum condition and Hadamard form for vectorvalued quantum fields in curved space-time}, Rev. Math. Phys.13 (2001) 1203-1246


\bibitem[Ve94]{Verch} R. Verch, {\em Local definiteness, primarity and quasiequivalence of quasifree Hadamard quantum states
in curved spacetime},  Commun. Math. Phys. 160, 507-536 (1994)

\bibitem[Wa94]{Wald} R.M. Wald, {\em Quantum Field Theory in Curved Spacetime and Black Hole Thermodynamics}. Chicago Lectures in Physics. University Of Chicago Press (1994)


\end{thebibliography}
\end{document}